\documentclass[submission,copyright,creativecommons]{eptcs}
\providecommand{\event}{NCMA 2024} 

\usepackage{iftex}

\ifpdf
  \usepackage{underscore}         
  \usepackage[T1]{fontenc}        
\else
  \usepackage{breakurl}           
\fi

\usepackage{amsmath}
\usepackage{amssymb}
\usepackage{amsfonts}
\usepackage{amsthm}

\theoremstyle{definition}
\newtheorem{definition}{Definition}

\theoremstyle{plain}
\newtheorem{theorem}{Theorem}

\newtheorem{lemma}[theorem]{Lemma}
\newtheorem{question}{Question}

\DeclareSymbolFont{rsfscript}{OMS}{rsfs}{m}{n}
\DeclareSymbolFontAlphabet{\mathrsfs}{rsfscript}

\DeclareMathOperator{\Ker}{Ker}

\usepackage{pstricks,pst-node,pst-text,pst-3d}
\usepackage{color}

\newcommand{\mA}{\mathrsfs{A}}

\newcommand{\mM}{\mathrsfs{M}}
\newcommand{\gR}{\mathrel{\mathfrak{R}}}
\newcommand{\gL}{\mathrel{\mathfrak{L}}}
\newcommand{\gD}{\mathrel{\mathfrak{D}}}

\newcommand{\sa}{synchronizing automata}
\newcommand{\san}{synchronizing automaton}
\newcommand{\sDFAs}{synchronizing DFAs}
\newcommand{\sDFA}{synchronizing DFA}

\usepackage{mathptmx}      
\usepackage{pgfplots}
\pgfplotsset{compat=1.10}

\title{Winning Strategies for the Synchronization Game\\ on Subclasses of Finite Automata\thanks{Mikhail Volkov was supported by the Ministry of Science and Higher Education of the Russian Federation, project FEUZ-2023-2022.}}

\def\titlerunning{Winning Strategies for the Synchronization Game}

\author{Henning Fernau \qquad\qquad  Carolina Haase\institute{Universit\"at Trier, Fachbereich IV, Informatikwissenschaften, Trier, Germany}\email{fernau@uni-trier.de  \qquad\qquad haasec@uni-trier.de}  \and Stefan Hoffmann\email{hoffmanns.tcs@gmail.com}  \and Mikhail Volkov\institute{Institute of Natural Sciences and Mathematics, Ural Federal University, Ekaterinburg, Russia}\email{m.v.volkov@urfu.ru}}

\def\authorrunning{H. Fernau, C. Haase, S. Hoffmann, M. Volkov}

\begin{document}
\maketitle

\begin{abstract}
We exhibit a winning strategy for Synchronizer in the synchronization game on every \san\ in whose transition monoid the regular $\gD$-classes form subsemigroups.
\end{abstract}

\section{Introduction}
\label{intro}

A complete deterministic finite automaton (DFA) is a pair $(Q,\Sigma)$ of two finite sets equipped with a map $Q\times\Sigma\to Q$ whose image at $(q,a)\in Q\times\Sigma$ is denoted by $q{\cdot}a$. We call $Q$ the \emph{state set} and $\Sigma$ the \emph{input alphabet}. Elements of $Q$ and $\Sigma$ are referred to as \emph{states} and, respectively, \emph{letters}, and for a state $q\in Q$ and a letter $a\in\Sigma$, we refer to $q{\cdot}a$ as the result of the \emph{action of $a$ at} $q\in Q$. The action of letters in $\Sigma$ naturally extends to the action of words over $\Sigma$: if $w=a_1a_2\cdots a_n$ with $a_1,a_2,\dots,a_n\in\Sigma$, then $q{\cdot}w:=(\dots((q{\cdot}a_1){\cdot}a_2)\dots){\cdot}a_n$.

A DFA $(Q,\Sigma)$ is called \emph{synchronizing} if there exists a word $w$ over $\Sigma$ whose action brings the DFA to one particular state no matter at which state $w$ is applied: $q{\cdot}w=q'{\cdot}w$ for all $q,q'\in Q$. Any word $w$ with this property is said to be a \emph{reset} word for the automaton.

Synchronizing automata serve as transparent and natural models of error-re\-sis\-tant systems in many applications (coding theory, robotics, testing of reactive systems) and reveal interesting connections with symbolic dynamics, substitution systems, and other parts of mathematics. We refer the reader to chapter~\cite{KV} of the `Handbook of Automata Theory' and survey \cite{Vo22} for an introduction to the area and an overview of its state-of-the-art.

The fourth-named author initiated viewing \sa{} through the lens of game theory; the motivation for this came from a game-theoretical approach to software testing suggested in~\cite{BGNV06}. In a synchronization game on a~DFA $\mathrsfs{A}$, two players, Alice (Synchronizer) and Bob (Desynchronizer), take turns choosing letters from the input alphabet of $\mA$. Alice who wants to synchronize $\mathrsfs{A}$ wins when the sequence of chosen letters forms a reset word. Bob aims to prevent synchronization or, if synchronization is unavoidable, to delay it as long as possible. Provided that both players play optimally, the outcome of such a game depends on the automaton only. This raises the problem of classifying \sa{} into those on which Alice and, respectively, Bob have a winning strategy. DFAs on which Alice can ensure win are of interest because they are more amenable to synchronization, in a sense. For brevity, we call such DFAs \emph{A-automata}.

A few initial results on synchronization games were obtained in~\cite{FMV}. In particular, \cite[Theorem 4]{FMV} provides an algorithm that, given a DFA $\mathrsfs{A}$ with $n$ states and $k$ input letters, decides who has a winning strategy in the synchronization game on $\mA$ in $O(n^2k)$ time. Thus, for any individual DFA, one can determine whether it is an A-automaton. Here, however, we are interested in general conditions ensuring that all synchronizing DFAs of a certain type are A-automata. One such condition was mentioned in~\cite{FMV}: Alice always wins on definite DFAs introduced in~\cite{PRS:1963}.

In \cite{FUN22}, the first three authors of the present note showed that within two further families of automata considered in the literature---weakly acyclic DFAs and commutative DFAs---every \sDFA\ is an A-automaton. Here we continue this line of research by designing a winning strategy for Alice that applies to \sa\ from yet another family of DFAs. Automata in this family are distinguished by a structure feature of their transition monoids: the regular $\gD$-classes in these monoids form subsemigroups. The set $\mathbf{DS}$ of all finite semigroups with regular $\gD$-classes being subsemigroups plays a distinguished role in the algebraic theory of regular languages; see \cite[Chapter 8]{Almeida-95}. Therefore, DFAs with transition monoids in $\mathbf{DS}$ often show up in the literature; see, e.g., \cite{AMSV09,AlSt09}. As they seem to have no specific name so far, we coin them DS-automata. Thus, our main result says that Alice can win the synchronization game on every synchronizing DS-automaton. Since the family of DS-automata is extensive and encompasses all the families above (definite, weakly acyclic, and commutative DFAs), this provides a vast generalization of the mentioned results from \cite{FMV,FUN22}.

Our approach is algebraic as it exploits the structural properties of transition monoids. We have collected all necessary prerequisites from semigroup theory in Section~\ref{prelim} to make the note self-contained, to a reasonable extent. Section~\ref{main} presents Alice's winning strategy in synchronization games on synchronizing DS-automata. In Section~\ref{discussion}, we relate our result to previously found facts about winning strategies for synchronization games and state two open questions.

\section{Preliminaries}
\label{prelim}

\subsection{Transition Monoids and Synchronization}
\label{subsec:transmonoid}

For a DFA $\mathrsfs{A}=(Q,\Sigma)$, the map $\tau_a\colon Q\to Q$ defined by the rule $q\mapsto q{\cdot}a$ is a transformation on the set $Q$.
\begin{definition}
\label{def:transmonoid}
The \emph{transition monoid} of a DFA $\mA=(Q,\Sigma)$ is the submonoid of the monoid of all transformations on the set $Q$ generated by the set $\{\tau_a\mid a\in\Sigma\}$.
\end{definition}

We denote the transition monoid of a DFA $\mA=(Q,\Sigma)$ by $T(\mA)$. It is easy to see that any product $\tau_{a_1}\tau_{a_2}\cdots\tau_{a_n}$ is nothing but the transformation $\tau_w$ defined by the rule $q\mapsto q{\cdot}w$ where $w$ stands for the word $a_1a_2\cdots a_n$. Thus, the transition monoid $T(\mA)$ can alternatively be defined as the monoid of all transformations on the set $Q$ caused by the action of words over $\Sigma$.

If $\mathrsfs{A}=(Q,\Sigma)$ is a synchronizing DFA and $w$ is a reset word for $\mA$, then the transformation $\tau_w$ is a constant map on $Q$, that is, $Q\tau_w=\{q\}$ for a certain $q\in Q$. Thus, the transition monoid of a \san{} always contains a constant transformation. Conversely, if $\zeta\in T(\mA)$ is a constant transformation, then any word $w$ with $\tau_w=\zeta$ is a reset word for $\mA$ and so $\mA$ is synchronizing. We see that synchronization is actually a property of the transition monoid of an automaton rather than the automaton itself: for DFAs $\mathrsfs{A}=(Q,\Sigma)$ and $\mathrsfs{A}'=(Q,\Sigma')$ with the same state set but different input alphabets, the equality $T(\mA)=T(\mA')$ guarantees that $\mA'$ is synchronizing if and only if so is $\mA$.

We can say a bit more about transition monoids of \sa{}, but for this, we first need to recall some concepts of semigroup theory. A non\-empty subset $I$ of a semigroup $S$ is called an \emph{ideal} in~$S$ if, for all $s\in S$ and $i\in I$, both products $si$ and $is$ lie in $I$. If $I$ and $J$ are two ideals in $S$, their intersection $I\cap J$ contains any product $ij$ with $i\in I$, $j\in J$. So $I\cap J$ is nonempty, and then it is easy to verify that  $I\cap J$ is again an ideal in $S$. Therefore, each finite semigroup $S$ has a least ideal called the \emph{kernel} of $S$ and denoted $\Ker S$.

The following observation is folklore but we provide a proof for the sake of completeness.
\begin{lemma}
\label{lem:kernel}
A DFA is synchronizing if and only if the kernel of its transition monoid consists of constant transformations.
\end{lemma}

\begin{proof}
The `if' part readily follows from the already mentioned fact that if the transition monoid of a DFA contains a constant transformation, then the DFA is synchronizing.

For the `only if' part, let $\mA$ be a \sDFA; then the set $C$ of all constant transformations in the transition monoid $T(\mA)$ is nonempty. For any $\tau\in T(\mA)$ and any $\zeta\in C$, we have
\begin{equation}\label{eq:rightzero}
\tau\zeta=\zeta.
\end{equation}
Equality~\eqref{eq:rightzero} implies that the set $C$ is contained in every ideal of the monoid $T(\mA)$, in particular, in its kernel $\Ker T(\mA)$. Further, for every $\tau\in T(\mA)$ and every $\zeta\in C$, the product $\zeta\tau$ is a constant transformation: if $Q\zeta=\{q\}$, then $Q\zeta\tau=\{q\tau\}$. Together with \eqref{eq:rightzero}, this observation implies that $C$  forms an ideal in $T(\mA)$. Since $\Ker T(\mA)$ contains $C$, we have $\Ker T(\mA)=C$ by the definition of the kernel.
\end{proof}

\subsection{Structure of Semigroups in $\mathbf{DS}$}

Green \cite{Green:1951} defined five important relations on every semigroup $S$, collectively referred to as \emph{Green's relations}, of which we need the following three:
\begin{itemize}
\item[] $a\gR b \Longleftrightarrow {}$ either $a=b$ or $a=bs$ and $b=at$ for some $s,t\in S$;
\item[] $a\gL \,b \Longleftrightarrow {}$ either $a=b$ or $a=sb$ and $b=ta$ for some $s,t\in S$;
\item[] $a\gD b \Longleftrightarrow {}$ $a\gR c$ and $c\gL b$ for some $c\in S$.
\end{itemize}
The relations $\gR$ and $\gL$ are obviously equivalencies. The definition of ${\gD}$ means that ${\gD}$ is the product of $\gR$ and $\gL$ as binary relations. As observed in \cite{Green:1951}, ${\gD}$ is also the product of $\gL$ and $\gR$, and this implies that $\gD$ is the least equivalence containing both $\gR$ and $\gL$.

An element $a$ of a semigroup $S$ is said to be \emph{regular} if $asa=a$ for some $s\in S$. A $\gD$-class $D$ is called \emph{regular} if it contains a regular element. (In this case, every element of $D$ is known to be regular; see \cite[Theorem 6]{Green:1951}.) We denote by $\mathbf{DS}$ the set of all finite semigroups $S$ such that the regular $\gD$-classes of $S$ are subsemigroups in $S$.

The structure of semigroups $\mathbf{DS}$ is well understood in terms of their decompositions into some basic blocks. As we will use this structural result, we recall the notions involved.

A \emph{semilattice} is a semigroup satisfying the laws of commutativity $xy=yx$ and idempotency $x^2=x$.
\begin{definition}
\label{def:semilattice}
Let $Y$ be a semilattice and $\{S_y\}_{y\in Y}$ a family of disjoint semigroups indexed by the elements of $Y$. A semigroup $S$ is said to be a \emph{semilattice of semigroups} $S_y$, $y\in Y$, if:
\begin{description}
  \item[{\normalfont(S1)}] $S=\bigcup_{y\in Y}S_y$;
  \item[{\normalfont(S2)}] each $S_y$ is a subsemigroup in $S$;
  \item[{\normalfont(S3)}] for every $y,z\in Y$ and every $s\in S_y$, $t\in S_z$, the product $st$ belongs to $S_{yz}$.
\end{description}
\end{definition}

We say that a semigroup $S$ is $m$-\emph{nilpotent over its kernel} ($m$ being a positive integer) if every product of $m$ elements of $S$ belongs to $\Ker S$. We call a semigroup \emph{nilpotent over its kernel} if it is $m$-nilpotent over its kernel for some $m$. (To a semigroupist, finite semigroups nilpotent over their kernels are familiar as finite \emph{Archimedean} semigroups.)

The following is a specialization of the equivalence (4c) $\Leftrightarrow$ (1b) in \cite[Theorem 3]{Shevrin:94} to finite semigroups\footnote{Theorem 3 in \cite{Shevrin:94} deals with semigroups in which every element has a power that belongs to a subgroup. Every finite semigroup has this property.}.
\begin{lemma}
\label{lem:ds}
Every semigroup in\/ $\mathbf{DS}$ is a semilattice of semigroups nilpotent over their kernels.
\end{lemma}

\section{Winning Strategy in Synchronization Games on DS-Automata}
\label{main}

We start with a visual yet rigorous description of the synchronization game under consideration. In this game, two players, Alice and Bob, play on a fixed DFA $\mathrsfs{A}=(Q,\Sigma)$. At the start, each state in $Q$ holds a token. During the game, some tokens can be removed according to the rules specified in the next paragraph. Alice wins if only one token remains, while Bob wins if he can keep at least two tokens unremoved for an indefinite amount of time.

Alice moves first, and then players alternate moves. The player whose turn is to move proceeds by selecting a letter $a\in\Sigma$. Then, for each state $q\in Q$ that held a token before the move, the token advances to the state $q{\cdot}a$. (In the standard graphical representation of $\mathrsfs{A}$ as the labelled digraph with $Q$ as the vertex set and the labelled edges of the form $q\xrightarrow{a}q{\cdot}a$, one can visualize the move as follows: all tokens simultaneously slide along the edges labelled $a$.) If several tokens arrive at the same state after this, all of them but one are removed so that when the move is completed, each state holds at most one token.

To illustrate, let us look at how the synchronization game might play out on the following DFA in which, initially, each state holds a token (shown in gray).
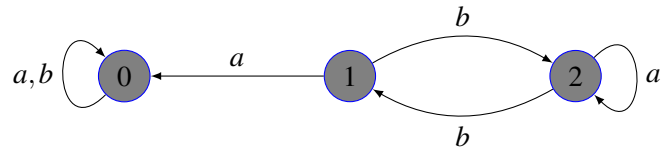
\begin{center}
\begin{tikzpicture}
	\node[fill=gray, circle, draw=blue, scale=1] (0) {$0$};
	\node[fill=gray, circle, draw=blue, scale=1, right of = 0, xshift= 2cm] (1) {$1$};
    \node[fill=gray, circle, draw=blue, scale=1, right of = 1, xshift= 2cm] (2) {$2$};
	\draw
		(1) edge[-latex, above]  node{$a$} (0)
        (1) edge[-latex, bend left, above]  node{$b$} (2)
		(2) edge[-latex, bend left, below] node{$b$}(1)
		(0) edge[-latex, loop left, left, in =135, out = -135, distance = 30] node{$a,b$} (0)
		(2) edge[-latex, loop right, right, out=45, in = -45, distance = 30] node{$a$} (2);
\end{tikzpicture}
\end{center}
If Alice chooses the letter $a$ on her first move, the tokens in states 0 and 2 remain due to the loops at these states. The token from state 1 moves to 0 and then is removed because the state 0 is `occupied'. Hence, the position after Alice's first moves looks as follows:
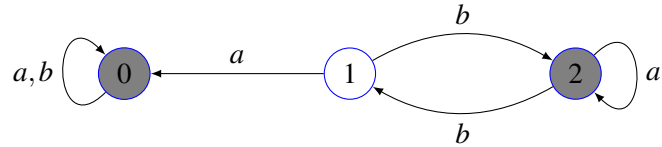
\begin{center}
\begin{tikzpicture}
	\node[fill=gray, circle, draw=blue, scale=1] (0) {$0$};
	\node[fill=white, circle, draw=blue, scale=1, right of = 0, xshift= 2cm] (1) {$1$};
    \node[fill=gray, circle, draw=blue, scale=1, right of = 1, xshift= 2cm] (2) {$2$};
	\draw
		(1) edge[-latex, above]  node{$a$} (0)
        (1) edge[-latex, bend left, above]  node{$b$} (2)
		(2) edge[-latex, bend left, below] node{$b$}(1)
		(0) edge[-latex, loop left, left, in =135, out = -135, distance = 30] node{$a,b$} (0)
		(2) edge[-latex, loop right, right, out=45, in = -45, distance = 30] node{$a$} (2);
\end{tikzpicture}
\end{center}
If Bob replies by choosing the letter $b$, the token in state 0 remains while the token from state 2 moves to 1. Here is the position after Bob's reply:
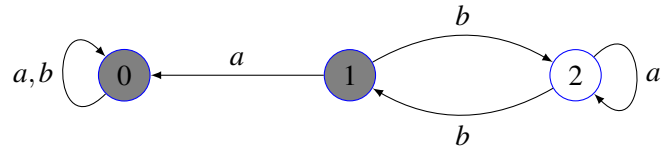
\begin{center}
\begin{tikzpicture}
	\node[fill=gray, circle, draw=blue, scale=1] (0) {$0$};
	\node[fill=gray, circle, draw=blue, scale=1, right of = 0, xshift= 2cm] (1) {$1$};
    \node[fill=white, circle, draw=blue, scale=1, right of = 1, xshift= 2cm] (2) {$2$};
	\draw
		(1) edge[-latex, above]  node{$a$} (0)
        (1) edge[-latex, bend left, above]  node{$b$} (2)
		(2) edge[-latex, bend left, below] node{$b$}(1)
		(0) edge[-latex, loop left, left, in =135, out = -135, distance = 30] node{$a,b$} (0)
		(2) edge[-latex, loop right, right, out=45, in = -45, distance = 30] node{$a$} (2);
\end{tikzpicture}
\end{center}
Now choosing $a$, Alice wins because after the token from state 1 moves to 0, it is removed, and we get the position with only one token:
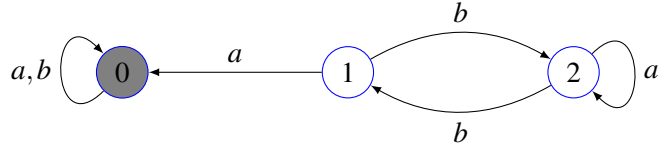
\begin{center}
\begin{tikzpicture}
	\node[fill=gray, circle, draw=blue, scale=1] (0) {$0$};
	\node[fill=white, circle, draw=blue, scale=1, right of = 0, xshift= 2cm] (1) {$1$};
    \node[fill=white, circle, draw=blue, scale=1, right of = 1, xshift= 2cm] (2) {$2$};
	\draw
		(1) edge[-latex, above]  node{$a$} (0)
        (1) edge[-latex, bend left, above]  node{$b$} (2)
		(2) edge[-latex, bend left, below] node{$b$}(1)
		(0) edge[-latex, loop left, left, in =135, out = -135, distance = 30] node{$a,b$} (0)
		(2) edge[-latex, loop right, right, out=45, in = -45, distance = 30] node{$a$} (2);
\end{tikzpicture}
\end{center}

Notice that Alice won the game described above only because of Bob's unfortunate reply. In fact, Bob has a winning strategy in the synchronization game on this DFA: if he repeats Alice's moves, that is, chooses the same letter Alice chose on her previous move, he can maintain two tokens unremoved forever. Hence, there are simple synchronizing DFAs that are not A-automata.

Recall that a DS-automaton is a DFA whose transition monoid lies in the set $\mathbf{DS}$ of all finite semigroups with regular $\gD$-classes being subsemigroups. The following is the main result of this note.

\begin{theorem}
\label{thm:main}
Alice has a winning strategy on every synchronizing DS-automaton.
\end{theorem}

\begin{proof}
Take an arbitrary synchronizing DS-automaton $\mA=(Q,\Sigma)$. We denote the transition monoid $T(\mA)$ by $S$ to lighten the notation. Since $S\in\mathbf{DS}$, by Lemma~\ref{lem:ds} there is a semilattice $Y$ such that $S$ is a semilattice of semigroups $S_y$, $y\in Y$, where each semigroup $S_y$ is nilpotent over its kernel.

The relation $\le$ defined by $x\le y\Longleftrightarrow  xy=x$ is known (and easy to see) to be a partial order on every semilattice, and so on $Y$. Due to the laws of commutativity and idempotency, the inequalities $xy\le x$ and $xy\le y$ hold for all $x,y\in Y$. Since the semilattice $Y$ is finite, it has a least element with respect to this order; denote it by $z$. Then $yz=zy=z$ for every $y\in Y$ whence the semigroup $S_z$ is an ideal in $S$ by item (S3) in Definition~\ref{def:semilattice}. Therefore $S_z$ contains the kernel $\Ker S$ of $S$, and therefore, $\Ker S_z\subseteq\Ker S$. (In fact, it is easy to show that the equality $\Ker S_z=\Ker S$ holds, but it is not needed for the present proof.)

Fix a positive integer $m$ such that the semigroup $S_z$ is $m$-nilpotent over its kernel $\Ker S_z$. We show that Alice can win in the synchronization game on $\mA$, using the following $m$-round strategy. Denote by $a^k_i$ and $b^k_i$ the $i$-th letters chosen in the $k$-th round by Alice and Bob, respectively. In each round, Alice chooses the first letter $a^{k}_1$ at random. Then, after each reply of Bob, she checks whether the word $u^k_i:=a^{k}_1b^{k}_1\cdots a^k_ib^k_i$ causes a transformation in $S_z$. If yes, then Alice starts the next round. If no, the transformation caused by $u^k_i$ lies in some subsemigroup $S_y$ with $y\ne z$ (recall that $S=\bigcup_{y\in Y}S_y$ by item (S1) in Definition~\ref{def:semilattice}). For each letter $a\in\Sigma$, denote by $y(a)$ the element of the semilattice $Y$ such that the transformation $\tau_a$ lies in the subsemigroup $S_{y(a)}$. Take any transformation $\tau\in S_z$. By Definition~\ref{def:transmonoid}, $\tau=\tau_{a_1}\tau_{a_2}\cdots\tau_{a_n}$ for some $a_1,a_2,\dots,a_n\in\Sigma$ whence by item (S3) in Definition~\ref{def:semilattice}, $z=y(a_1)y(a_2)\cdots y(a_n)$. If $y\le y(a_i)$ for all $i=1,2,\dots,n$, then $y\le y(a_1)y(a_2)\cdots y(a_n)=z$, and this would contradict the choice of $z$ as the least element with respect to $\le$ and the assumption $y\ne z$. Thus, there must be a letter $a\in\Sigma$ such that $y\nleq y(a)$, and Alice chooses any such letter $a$ as $a^{k}_{i+1}$.

By the construction, if $S_y\ne S_z$, then the index $x$ of the subsemigroup $S_x$ containing the transformation caused by the word $u^k_{i+1}$ is strictly less than $y$ in the partially ordered set $(Y;{\le})$. Hence, for each $k$, the number $\ell_k$ of pairs of moves in the $k$-th round does not exceed the maximum length of strictly decreasing chains in $(Y;{\le})$, and the transformation caused by the word $u^k_{\ell_k}$ lies in the subsemigroup $S_z$. Then the transformation $\tau_w$ caused by the word $w:=u^1_{\ell_1}u^2_{\ell_2}\cdots u^m_{\ell_m}$ is a product of $m$ elements of $S_z$ and so $\tau_w$ belongs to $\Ker S_z$ as the semigroup $S_z$ is $m$-nilpotent over its kernel. Since $\mA$ is a \san, the kernel $\Ker S$ of its transition monoid consists of constant transformations by Lemma~\ref{lem:kernel}. From the inclusion $\Ker S_z\subseteq\Ker S$ registered above, we conclude that $\tau_w$ is a constant transformation, and so $w$ is a reset word for $\mA$.
\end{proof}

\section{Relations to Earlier Results and Future Work}
\label{discussion}

\subsection{Corollaries}

In the introduction, we mentioned a few previously known families of A-autom\-a\-ta. Now we show that all these families consist of DS-automata so their winning strategies are subsumed by that of Theorem~\ref{thm:main}.

\paragraph{Definite automata.} This DFA family was introduced by some of the pioneers of automata theory back in 1963~\cite{PRS:1963}. In~\cite{PRS:1963}, the term `automaton' meant a recognizer, that is, a DFA with a designated initial state and a distinguished set of final states. However, DFAs without initial and final states as defined in the present note also appeared in~\cite{PRS:1963} but under the name `transition tables'. The following is \cite[Definition 13]{PRS:1963} stated in our terminology and notation.
\begin{definition}
\label{def:definite}
A DFA $(Q,\Sigma)$ is \emph{weakly $k$-definite} if for every word $w$ of length at least $k$ over $\Sigma$, $q{\cdot}w=q'{\cdot}w$ for all $q,q'\in Q$. A DFA is $k$-\emph{definite} if it is weakly $k$-definite but not weakly $(k-1)$-definite. A DFA is \emph{definite} if it is $k$-definite for some $k$.
\end{definition}

By Definition~\ref{def:definite}, every definite DFA is synchronizing and any word of length at least $k$ is a reset word for every $k$-definite DFA. Therefore, Alice wins on every definite automaton by selecting her moves at random.

The transition monoid of a definite DFA is nilpotent over its kernel. This fact is implicitly contained in~\cite{PRS:1963} and in the explicit form, it is a part of~\cite[Theorem~3]{Zal72}. Comparing it with Lemma~\ref{lem:ds}, we see that definite automata constitute a special subfamily of DS-automata. Moreover, for definite automata, Alice's winning strategy from Theorem~\ref{thm:main} specializes exactly to the random choice of moves. In fact, a DFA is definite if and only if Alice can win by pure random choices of moves.

\paragraph{Commutative automata.} A DFA is said \emph{commutative} if its transition monoid is commutative, that is, satisfy the law $xy=yx$. Synchronizing commutative automata were considered in \cite{Rys96,Rys97,Hof21a}. A simple winning strategy for Alice in synchronization games on such automata was suggested in~\cite[Theorem 5.2]{FUN22}: Alice must just choose letters spelling a reset word. Hence, as in the previous case, Alice can completely ignore the moves of Bob.

Obviously, on any commutative semigroup, Green's relations $\gR$ and $\gL$ coincide with each other, and hence, with the relation $\gD$. If an element $a$ of a commutative semigroup $S$ is regular, then $a=a^2s$ for some $s\in S$ whence  $a\gR a^2$. By \cite[Theorem 7]{Green:1951}, this ensures that the $\gD$-class containing $a$ is a subsemigroup (and even a subgroup) of $S$. Thus, in all commutative semigroups, regular $\gD$-classes are subsemigroups. Therefore, commutative DFAs are DS-automata, and Theorem~\ref{thm:main} applies to commutative \sDFAs.

\paragraph{Weakly acyclic automata.} Let $\mA = (Q,\Sigma)$ be a DFA and $p, q \in Q$. We say that $q$ is \emph{reachable from~$p$ in} $\mA$ if either $p=q$ or there exists a word $w$ over $\Sigma$ such that $q = p{\cdot}w$. The reachability relation of any DFA is reflexive and transitive. A DFA is called \emph{weakly acyclic} if its reachability relation is a partial order.

Various properties of synchronizing weakly acyclic DFAs were considered in \cite{Ryz19,Hof21b}. A winning strategy for Alice in synchronization games on such automata was suggested in~\cite[Theorem 2.3]{FUN22}.

By~\cite[Proposition 6.2]{BF80} the transition monoid $T(\mA)$ of any weakly acyclic DFA $\mA$ is $\gR$-\emph{trivial}, which means that Green's relation $\gR$ on $T(\mA)$ coincides with the equality relation. It is well known that regular $\gD$-classes are subsemigroups in any $\gR$-trivial semigroup, but we failed to locate a source where this fact was formulated such that it would be convenient to refer to it. Therefore we state it as a lemma and provide a proof.
\begin{lemma}
\label{lem:rtrivial}
In every $\gR$-trivial semigroup, regular $\gD$-classes are subsemigroups.
\end{lemma}

\begin{proof}
Let $S$ be an $\gR$-trivial semigroup and let $D$ be its regular $\gD$-class. Every element $b\in D$ is regular, so that $b=btb$ for some $t\in S$. Then $b\gR bt$ whence $b=bt$ since $S$ is $\gR$-trivial. We have $b^2=btbt=(btb)t=bt=b$. Now if $a\in D$, then $a\gL b$ since ${\gD}={\gL}$ in every $\gR$-trivial semigroup. By the definition of the relation $\gL$, we have either $a=b$ or $a=sb$ for some $s\in S$. If $a=b$, then $ab=b^2=b=a$. If $a=sb$,  multiplying through on the right by $b$, we get $ab=sb^2=sb=a$. Thus, $ab=a\in D$ for arbitrary $a,b\in D$, whence $D$ is a subsemigroup.
\end{proof}
Thus, weakly acyclic DFAs are DS-automata, and Theorem~\ref{thm:main} applies again.

\subsection{Open Questions}

Theorem~\ref{thm:main} generalizes several earlier results on A-automata. Can it be further generalized? This is an interesting question, but it requires specification.

We mentioned in Section~\ref{subsec:transmonoid} that synchronization is a property of transition monoids. This is not true for the property of being an A-automaton. Moreover, for every \sDFA{} $\mathrsfs{A}=(Q,\Sigma)$, there exists an A-automaton $\mathrsfs{A}'=(Q,\Sigma')$  such that $T(\mA)=T(\mA')$. To see this, let $\Sigma':=\Sigma\cup\{c\}$ where the action of the added letter $c$ coincides with the action of a fixed reset word for $\mA$. The  transformations caused by the letters in $\Sigma'$ generate the same submonoid of the monoid of all transformations on the set $Q$ as do the transformations caused by the letters in $\Sigma$. Still, Alice instantly wins the synchronization game on $\mA'$ by choosing the letter $c$ on her first move.

Thus, one cannot hope for a characterization of A-automata in terms of transition monoids. One can, however, try to find new sets $\mathbf{P}$ of finite semigroups such that $\mathbf{DS}\subsetneq\mathbf{P}$ and every \sDFA{} whose transition monoid lies in $\mathbf{P}$ is an A-automaton.

It is easy to verify that the property of being an A-automaton is inherited by subautomata, homomorphic images, and finite direct products. This suggests looking for sets $\mathbf{P}$ closed under corresponding operations with semigroups. A set of finite semigroups closed under forming finite direct products and taking subsemigroups and homomorphic images is called a \emph{pseudovariety}; the set $\mathbf{DS}$ is an example of a pseudovariety.
Using the notion of a pseudovariety, we can specify a possible direction towards generalizing Theorem~\ref{thm:main} as follows:
\begin{question}
\label{que:pseudo}
Is there a pseudovariety\/ $\mathbf{P}$ of finite semigroups that strictly contains $\mathbf{DS}$ while all \sDFAs{} with transition monoids in $\mathbf{P}$ are A-automata?
\end{question}

The 5-element \emph{Brandt semigroup} $B_2$ consists of the following five $2\times2$-matrices, multiplied according to the usual rule:
\[
\begin{pmatrix}
1 & 0\\ 0 & 0
\end{pmatrix},\
\begin{pmatrix}
0 & 1\\ 0 & 0
\end{pmatrix},\
\begin{pmatrix}
0 & 0\\ 1 & 0
\end{pmatrix},\
\begin{pmatrix}
0 & 0\\ 0 & 1
\end{pmatrix},\
\begin{pmatrix}
0 & 0\\ 0 & 0
\end{pmatrix}.
\]
The monoid $B_2^1$ obtained by adding the identity $2\times2$-matrix to $B_2$ is the syntactic monoid of the language $(ab)^*$, and hence, the transition monoid of the minimal automaton $\mM$ of this language (shown below).
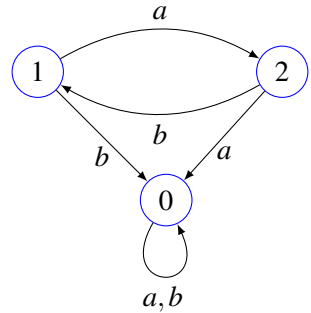
\begin{center}
\begin{tikzpicture}
	\node[fill=white, circle, draw=blue, scale=1] (0) {$0$};
	\node[fill=white, circle, draw=blue, scale=1, above left of = 0, xshift= -1cm, yshift= 1cm] (1) {$1$};
    \node[fill=white, circle, draw=blue, scale=1, right of = 1, xshift= 2.25cm] (2) {$2$};
	\draw
		(1) edge[-latex, below]  node{$b$} (0)
        (1) edge[-latex, bend left, above]  node{$a$} (2)
		(2) edge[-latex, bend left, below] node{$b$}(1)
		(2) edge[-latex, below]  node{$a$} (0)
		(0) edge[-latex, loop below, below, in =-65, out = -120, distance = 30] node{$a,b$} (0);
		\end{tikzpicture}
\end{center}
It is known that every pseudovariety of finite semigroups not contained in $\mathbf{DS}$ must include the semigroup $B_2$ (this fact occurs as Exercise~8.1.6 in \cite{Almeida-95}; the solution to this exercise follows from the proof of \cite[Theorem 3]{Margolis81}). Thus, if Bob had a winning strategy on the automaton $\mM$, then the answer to Question~\ref{que:pseudo} would be `No', and moreover, $\mathbf{DS}$ would be the largest pseudovariety $\mathbf{P}$ with the property that Alice can win the synchronization game on every \sDFA{} whose transition monoid lies in $\mathbf{P}$. However, it is easy to see that Alice wins on $\mM$: she can start with choosing $a$; if Bob replies with $a$, he loses, and if he replies with $b$, Alice wins by choosing $b$.

The fact that $\mM$ is an A-automaton, along with other examples of A-autom\-a\-ta beyond the family of DS-automata, indicates that Question~\ref{que:pseudo} might have an affirmative answer. We have a candidate pseudovariety to witness such an answer but its definition requires more structure theory of semigroups than assumed here.

Another question of interest concerns the speed of synchronization for A-automata. When we mentioned in the introduction that A-automata seem more amenable to synchronization, we meant that they tend to have short reset words. Indeed, in all examples we know, an A-automaton with $n$ states admits a reset word of length at most $n-1$. For DFAs in the range of Theorem~\ref{thm:main}, that is, synchronizing DS-automata, this was established in \cite[Theorem 2.6]{AlSt09}. The DFA $\mM$ with 3 states is reset by the words $a^2$ and $b^2$ of length 2 and hence provides another example. These observations lead to the next question.

\begin{question}
\label{que:length}
Is each A-automaton with $n$ states reset by a word of length $n-1$?
\end{question}

Recall that for each $n>2$, there exist \sDFAs\ with $n$ states whose shortest reset words have length $(n-1)^2$; see \cite[Lemma 1]{Cerny:1964}. It can be verified that none of these `slowly synchronizing' DFAs are A-automata.

\paragraph{Acknowledgement.} The authors are grateful to the anonymous referees for their attention and remarks.

\end{document}